\documentclass[conference]{IEEEtran}
\usepackage{graphicx}
\usepackage{mathrsfs}
\usepackage{graphicx}
\usepackage{amssymb}
\usepackage{indentfirst}
\usepackage{threeparttable}
\usepackage{booktabs}
\usepackage{subfigure}
\usepackage{graphicx}
\usepackage{algorithm,algorithmic}
\usepackage{multirow}
\usepackage{booktabs}
\usepackage{bm}
\usepackage{enumerate}
\usepackage{cite}
\usepackage{amsmath} % assumes amsmath package installed
\usepackage{amssymb}  % assumes amsmath package installed

\usepackage{amsmath,amssymb,amsfonts}
\usepackage{graphicx}
\usepackage{subfigure}
\usepackage{bm}

\newtheorem{property}{Property}

\begin{document}

\title{Distributed Compressive Sensing Based  Doubly Selective  Channel Estimation for  Large-Scale  MIMO Systems}

%\author{\IEEEauthorblockN{Bo Gong\IEEEauthorrefmark{1}, Qibo Qin\IEEEauthorrefmark{1}, Xiang Ren\IEEEauthorrefmark{1},  Lin Gui\IEEEauthorrefmark{1}, Hanwen Luo\IEEEauthorrefmark{1}, Wen Chen\IEEEauthorrefmark{1}}\\
%\IEEEauthorblockA{\IEEEauthorrefmark{1}Dept. of Electronic Engineering, Shanghai Jiao Tong University, Shanghai, China\\
%\IEEEauthorrefmark{1}Email: \{gongbo, qinqibo, rex, guilin,  hwluo, wenchen\}@sjtu.edu.cn\\
%}

\author{\IEEEauthorblockN{ Bo~Gong, Qibo~Qin, Xiang~Ren, Lin~Gui, Hanwen~Luo and Wen~Chen}

\IEEEauthorblockA{Department of Electronic Engineering, Shanghai Jiao Tong University, China\\
Email: \{gongbo;~ qinqibo;~ rex;~ guilin;~  hwluo;~ wenchen\}@sjtu.edu.cn}}

\maketitle
\begin{abstract}
Doubly selective (DS) channel estimation in large-scale multiple-input multiple-output (MIMO) systems is a  challenging problem due to the requirement of unaffordable pilot overheads and prohibitive complexity. In this paper, we propose a novel distributed compressive sensing (DCS) based channel estimation scheme to solve this problem. In the scheme, we introduce the basis expansion model (BEM) to reduce the required channel coefficients and  pilot overheads. And due to the common  sparsity of all the transmit-receive antenna pairs in delay domain, we estimate the BEM coefficients by considering the DCS framework,   which has a simple linear structure with low complexity. Further more, a linear smoothing method  is proposed to improve the estimation accuracy. Finally, we conduct various simulations to verify the validity of the proposed scheme and demonstrate the performance gains of  the proposed scheme compared with conventional schemes.
\end{abstract}

%\begin{IEEEkeywords}
%compressive sensing, channel estimation, doubly selective, massive MIMO.
%\end{IEEEkeywords}

\section{Introduction}

Large-scale  multiple-input multiple-output (MIMO) attracts much academic interest and  is considered as a promising technology in the incoming fifth-generation cellular systems \cite{Larsson2014}. It enhances the data throughput and improves the link reliability of wireless communication system by taking advantage of the spatial multiplexing gains. In order to benefit from large-scale MIMO, one must obtain accurate channel state information (CSI)  which guarantees data recovery and contributes to multi-antenna array gains.

 Time and frequency selective channel, which is also referred to as  doubly-selective (DS) channel,  is related to many wireless access links, such as  high-speed trains \cite{Ren2015} and millimeter-wave  communications \cite{Rangan2014}. Frequency selectivity is caused by multipath propagation and time selectivity results from Doppler shift. For the DS channel estimation in large-scale MIMO-OFDM systems, there exist a large number of transmitting antennas,  inter-carrier interference (ICI) resulting  from the time selectivity  and  the multipath effect caused by the frequency selectivity. All of them increase the number of channel coefficients to be estimated greatly, resulting in the requirement of unaffordable pilot overheads and prohibitive complexity.

 Most of the researchers adopt time division duplex (TDD) in large-scale MIMO systems. They  take advantage of the channel reciprocity between the uplink and downlink channels. The uplink channel estimation is relatively simple due to the limited number of antennas in mobile terminals.  And then the transposition of CSI from uplink training is utilized as the downlink CSI directly.   But it is not suited for DS channels for the reason that the uplink CSI may be outdated for the  time selectivity.  In our proposed scheme, we estimate the downlink CSI without the uplink channel information in TDD large-scale MIMO systems. To our best knowledge, little has been done about downlink DS channel estimation in large-scale  MIMO systems.

Compressive sensing (CS) is an important framework to lower the pilot overheads and complexity of the channel estimation by taking advantage of the channel sparsity.  In \cite{Masood2015}, the authors propose compressive estimation schemes for flat fading  channel in large-scale MIMO systems.  \cite{Rao2014}  proposes a compressive CSIT estimation scheme under frequency-selective channels for frequency division duplex (FDD) large-scale MIMO systems. \cite{Nan2015},\cite{Hou2014} consider the compressive frequency-selective channel estimation for TDD large-scale MIMO systems. \cite{Ren2015},\cite{Cheng2013}  present the researches about DS channel estimation based on CS and distributed compressive sensing (DCS) in single-input single-output (SISO) systems. \cite{Ren2013} proposes low coherence compressed (LCC) channel estimation scheme for DS channels in MIMO systems. However, this scheme utilizes the sparsity in delay-doppler domain which is reduced notably by  a large doppler shift and a large  number of antennas.

In this paper, a compressive channel estimation scheme for DS channel in large-scale MIMO systems  is proposed. In this scheme,  the basis expansion model (BEM) is introduced to reduce the number of coefficients to be estimated. The number of BEM coefficients is $D/N$ of the number of channel coefficients, in which $D$ is the BEM order and $N$
  is the number of subcarriers, $D \ll N$. As the number of the coefficients to be estimated is reduced, the required pilot overheads  decrease as well. Then we  analyze the common sparsity of the BEM coefficients between all the transmit-receive antenna pairs in delay domain.  The BEM coefficients estimation is formulated as a DCS problem, which has a linear structure with low complexity. Moreover,  a linear smoothing method in large-scale MIMO systems is proposed to  improve the accuracy of the estimation by reducing the modeling error.  Finally, simulation results verify the effectiveness of the proposed scheme and  show the performance gains compared with the conventional schemes.

\emph{Notations}: $\left(  \cdot  \right)^T$ denotes matrix transposition, $\left(  \cdot  \right)^H$ represents matrix conjugate transposition.  $diag(\cdot)$ means a diagonal matrix, $\left|  \cdot  \right|$ denotes the absolute value, $\left\langle { \cdot , \cdot } \right\rangle $ denotes the inner product, ${\left\|  \cdot  \right\|_2}$ stands for the Euclidean norm, ${\left\|  \cdot  \right\|_0}$ denotes the number of nonzero values.  $ \otimes $ represents Kronecker product,  $\mathcal{S}$ indicates a set, ${A\left[m,n\right]}$ represents  the $(m+1,n+1)$-th  element of matrix $\mathbf A$. ${{\left[\mathbf A\right]}_\mathcal{S}}$ represents the selected rows of $\mathbf A$, whose indices correspond to the set $\mathcal{S}$. $\mathcal{CN}(0,{\sigma ^2})$  represents the complex Gaussian distribution with zero mean and $\sigma ^2$ variance. $\mathbf{I}_x$ means the identity matrix of order $x$. $vec(\mathbf{A})$ denotes the column-ordered vectorization of matrix $\mathbf{A}$.

   .

\section{System model and Background}

\subsection{System model}
\subsubsection{Large-scale MIMO}

  We consider a TDD large-scale  MIMO-OFDM system, in which the base station  with a great many  antennas  serves a number of terminals with a single antenna. The antenna array is arranged in a square, which consists of ${N_t}$ = $Z$ $\times$ $Z$ antennas. For each terminal, the downlink transmission includes $N_t$ transmitting antennas and one receiving antenna.  In OFDM systems,  $N$ subcarriers are in a parallel transmission. A part  of the subcarriers  are selected as pilot subcarriers and the remaining ones transmit data. ${{\mathbf{s}}^{\left( {{n_t}} \right)}} \in {\mathbb{C}^{N \times 1}}$ represents the  OFDM symbol transmitted by the $n_t$-th antenna in time domain. Its corresponding form in frequency domain is ${{\mathbf{S}}^{\left( {{n_t}} \right)}} = {\bf{W}}{{\bf{s}}^{\left( {{n_t}} \right)}}$, in which $\mathbf{W}$ is the discrete fourier transform (DFT) matrix,   ${W\left[m,n\right]} = {N^{-1/2}}\exp \left( {-j2\pi mn/N} \right)$, $m,n \in \left[ {0,N-1} \right]$.

The channel in  time domain is assumed to be a finite impulse response (FIR) filter. $h^{({n_t})}\left[n,l\right]$ represents the channel coefficient  of the $(l+1)$-th tap at the $(n+1)$-th  instant of the channel between the $n_t$-th antenna and the terminal. We assume that  $h^{({n_t})}\left[n,l\right]=0$, for $l < 0$ and $l  \ge  L$, in which $L$ is the length of the channel. The element of the received vector ${\mathbf{y}} \in {\mathbb{C}^{N \times 1}}$ by a terminal can be  expressed as
\begin{equation}
y[n] = \sum\limits_{{n_t} = 1}^{{N_t}} {\sum\limits_{l = 0}^{L - 1} {h^{({n_t})}\left[n,l\right]{s^{({n_t})}}[n - l]} }  + e[n],
\end{equation}
 in which, $e[n] \sim \mathcal{CN}(0,{\sigma ^2})$, $(n \in [0,N-1])$, is the noise term. $\mathbf{H}_t^{({n_t})} \in {\mathbb{C}^{N \times N}}$ describes channel matrix in time domain,
 \begin{equation}
 {{H^{(n_t)}_t} \left[p,q\right]} = h^{(n_t)}[p,\bmod (p - q,N)], (p,q \in [0,N - 1]),
 \end{equation}
 and another expression of the received vector ${\mathbf{y}}$ is
 \begin{equation}
 {{\bf{y}}} = \sum\limits_{{n_t} = 1}^{{N_t}} {{\bf{H}}_t^{({n_t})}{{\bf{W}}^H}{\bf{s}}^{({n_t})}}  + {\bf{e}}.
 \end{equation}
  The channel matrix in frequency domain ${\mathbf{H}_f^{({n_t})}}$ can be derived from $\mathbf{H}_f^{({n_t})}$ $=$ $\mathbf{W}{\mathbf{H}_t^{({n_t})}}$$\mathbf{W}^H$. The received vector in frequency domain $\mathbf{Y} \in {\mathbb{C}^{N \times 1}}$ is

\begin{equation}\label{recfre}
\mathbf{Y}  =  \sum\limits_{{n_t} = 1}^{{N_t}}{\mathbf{H}_f^{({n_t})}}\mathbf{S}^{({n_t})} + \mathbf{E},
\end{equation}
 in which $\mathbf{E}$ is the noise term in the frequency domain.

\subsubsection{BEM}

BEM \cite{Tang2011} is an important technique for DS channel estimation, which can reduce the number of the channel coefficients  to be estimated.  Consider an OFDM system with $N$ subcarriers and $L$ channel taps. ${\mathbf{h}}_l^{(n_t)} {=} {( {
{h^{(n_t)}[0,l]}, \cdots ,{h^{(n_t)}[N - 1,l]}  } )^T}  \in {\mathbb{C}^{N \times 1}}$ denotes the $l$-th channel tap. Each ${\mathbf{h}^{(n_t)}_l}, l \in [0,L-1]$ can be expressed as
\begin{equation}
{{\bf{h}}^{(n_t)}_l} = {\bf{V}}{{\bm{\theta}}^{(n_t)}_l} + {{\bm{\varepsilon }}^{(n_t)}_l},
\end{equation}
 in which ${{\bm{\theta}}^{(n_t)}_l} = {
({\theta^{(n_t)}[0,l]},  {\theta^{(n_t)}[1,l]},  { \cdots },  {\theta^{(n_t)}[D - 1,l]})^T
} \in  {\mathbb{C}^{D \times 1}}$ is the BEM coefficients, ${\bm{\varepsilon}^{(n_t)} _l} = {({\varepsilon^{(n_t)} [0,l]}, {\varepsilon^{(n_t)} [1,l]},  \cdots, {\varepsilon ^{(n_t)}[N - 1,l]})^T}$ is the BEM modeling error. ${\bf{V}} {=} \left(
{{{\bf{v}}_0}}, {{{\bf{v}}_1}}, \cdots, {{{\bf{v}}_{D - 1}}} \right)$, ${\mathbf{v}_d}$  $(d \in [0,D-1])$ is the BEM basis function, and $D$ $(D \ll N)$ is the BEM order. Apparently the number of channel coefficients to be estimated is reduced from $NL$ to $DL$ for one transmit-receive antenna pair. The vector of the channel taps for the $n_t$-th transmitting antenna can be formulated as
\begin{equation}\label{recovery}
\overline {\bf{h}}^{(n_t)}  = \left( {{\bf{V}} \otimes {{\bf{I}}_L}} \right)\overline {\bm{\theta}}^{(n_t)}  + \overline {\bm{\varepsilon }}^{(n_t)} ,
\end{equation}
in which
\begin{equation}
\begin{split}
{\overline{\bf{h}}}^{(n_t)}  {=} {( {{{(\widetilde {\bf{h}}}^{(n_t)}_0)^T}, \cdots ,{{(\widetilde {\bf{h}}}^{(n_t)}_{N - 1})^T}})^T},  \\
{\overline {\bm{\theta}}}^{(n_t)}  {=} {( {{{(\widetilde {\bm{\theta}}}^{(n_t)}_0)^T}, \cdots ,{{(\widetilde {\bm{\theta}}}^{(n_t)}_{D {-} 1}}})^T)^T},   \\
 \overline {\bm{\varepsilon }}^{(n_t)}  {=} {( {(\widetilde {\bm{\varepsilon }}_0^{(n_t)})^T, \cdots ,(\widetilde {\bm{\varepsilon }}_{N - 1}^{(n_t)})^T})^T},
 \end{split}
\end{equation}
 and
 \begin{equation}  \nonumber
\begin{split}
{\widetilde {\bf{h}}^{(n_t)}_n} {=} {( {h^{(n_t)}[n,0],  \cdots,  h^{(n_t)}[n,L - 1])}^T} {\in} {\mathbb{C}^{L \times 1}},
n {\in} [0,N{-}1],  \\
{\widetilde {\bm{\theta}}^{(n_t)}_d} {=} {({\theta^{(n_t)}[d,0],  \cdots,  \theta^{(n_t)}[d,L - 1])}^T} {\in} {\mathbb{C}^{L \times 1}},
d {\in} [0,D{-}1],   \\
{\widetilde {\bm{\varepsilon }}^{(n_t)}_n} {=} {( {{\varepsilon }^{(n_t)}[n,0],  \cdots,  {\varepsilon }^{(n_t)}[n,L - 1]})^T} {\in} {\mathbb{C}^{L \times 1}},
n {\in} [0,N{-}1].
\end{split}
\end{equation}

By simple derivation, we know that the channel matrix in frequency domain can be expressed as
\begin{equation}\label{matfre}
 {{\bf{H}}^{(n_t)}_f} = \sum\limits_{d = 0}^{D - 1} {{{\bf{V}}_d}{{\bm{\Theta}}^{(n_t)}_d}}  + {{\bf{\Delta }}^{(n_t)}},
\end{equation}
 in which ${\mathbf{V}_d} = {\mathbf{W}}diag\left( {{\mathbf{v}_d}} \right)\mathbf{W}^H,$  ${\bm{\Theta}^{(n_t)}_d} = diag{({{\sqrt N} {\mathbf{W}}   {(
{(\widetilde{\bm{\theta}}^{(n_t)}_d)^T,} {{\mathbf{0}_{1 \times (N - L)}}
})}^T}),}$ ${\bm{\Delta}}$ is the modeling error \cite{Cheng2013}.

\subsection{CS and DCS}

CS is an attractive framework, which recovers a high-dimensional sparse signal from a low dimensional observed vector.  CS solves the  underdetermined problem
\begin{equation}
\mathbf{y} = \mathbf{A} \mathbf{x}+\mathbf{e},
\end{equation}
 in which $\mathbf{x} \in \mathbb{C}^{N \times 1}$ is an unknown vector with sparsity $K$, i.e.${\left\| \mathbf{x} \right\|_{{0}}} = K$.  $\mathbf{A} \in \mathbb{C}^{M \times N} (M \ll N)$  is the measurement matrix, $\mathbf{y} \in \mathbb{C}^{M \times 1}$ represents the observed vector, and $\mathbf{e}$ denotes the noise term. Fundamental research indicates that if $\mathbf{x}$ is a sparse vector and $\mathbf{A}$ satisfies restricted isometry  property (RIP) condition \cite{Duarte2011}, a high probability of exact recovery of $\mathbf{x}$ can be guaranteed. But it is difficult to  verify RIP condition for the prohibitive complexity. Instead, mutual coherence property (MCP) \cite{Duarte2011} is an important reference value of the measurement matrix. It is defined as
\begin{equation}
\mu(\mathbf{A})=\max_{1\leq i\neq j\leq N}\frac{\vert\left\langle \mathbf{a}_{i},\mathbf{a}_{j}\right\rangle \vert}{{\lVert\mathbf{a}_{i}\rVert}_{2}{\lVert\mathbf{a}_{j}\rVert}_{2}},
\end{equation}
where ${\bf a}_{i}$ and ${\bf a}_{j}$ denote the columns of ${\bf A}$.  As proved in \cite{Duarte2011}, a smaller MCP will lead to a more accurate recovery of $\mathbf{x}$. Basis pursuit (BP) \cite{Chen1998} and orthogonal matching pursuit (OMP) \cite{Mallat1993} are the widely adopted recovery algorithms of CS.

DCS framework is applied to jointly compress and recover  multiple correlated sparse  signals.  The basic form of DCS is
\begin{equation}
\mathbf{Y} = \mathbf{A}\mathbf{X} + \mathbf{E},
\end{equation}
 in which, $\mathbf{Y} = [{\mathbf{Y}_0}, \cdots ,{\mathbf{Y}_{J-1}}] \in {\mathbb{C}^{M \times J}}$,  $\mathbf{X} = [{\mathbf{X}_0}, \cdots , {\mathbf{X}_{J-1}}] \in {\mathbb{C}^{M \times J}}$, ${\mathbf{Y}_j}$ $ (j \in [0,J-1 ])$ and ${\mathbf{X}_j}$  $ ( j \in [0,J-1])$ denotes the $j$-th column of $\mathbf{Y}$ and $\mathbf{X}$ respectively.  All the columns of $\mathbf{X}$ share the same nonzero positions. $\mathbf{E} = [{\mathbf{E}_0}, \cdots ,{\mathbf{E}_{J - 1}}] \in {\mathbb{C}^{M \times J}}$ is the noise matrix. $\mathbf{A}$ is the measurement matrix as above. It has been proved that DCS provides higher accuracy  with fewer observed values than CS by utilizing the common sparsity. Simultaneous-OMP (SOMP) \cite{Duarte2011} is an important algorithm for the recovery of DCS.

\section{the proposed channel estimation scheme in Large-Scale MIMO Systems}

\subsection{Channel Sparsity  in Large-Scale MIMO Systems}

To  clearly elaborate the channel sparsity in large-scale MIMO systems, we introduce two properties of the channel as follows:

\begin{property}
The channel coefficients $\{\widetilde {\bf{h}}^{(n_t)}_n\}{\in} {\mathbb{C}^{L \times 1}}$ $(n \in [0,N-1])$ of a transmit-receive antenna pair have common sparsity in delay domain \cite{Cheng2013}, i.e., their nonzero positions are the same.
\end{property}

 Assume that there are $L$ channel taps  for a  transmit-receive antenna pair and the index set of them is denoted  as  $[0,L-1]$. There exist  $K$ $(K \ll L)$ strong taps, $\left\{ {l_1}, \cdots ,{l_K} \right\}\subset [0,L-1]$,    and the remaining ones are minor channel coefficients which can be neglected. It means that ${\mathbf{h}^{(n_t)}_l}=\textbf{0}$, $l \notin \left\{ {l_1}, \cdots ,{l_K} \right\}$. $\{\widetilde {\bf{h}}^{(n_t)}_n\}$  $(n \in [0,N-1])$ are all $K$-sparse vectors and their common non-zero positions are $\left\{ {l_1}, \cdots ,{l_K} \right\}$.

\begin{property}
In large-scale MIMO systems, all the transmit-receive links  are scattering  invariantly in space and share common sparsity in delay domain  if $\frac{{{d_{\max }}}}{C} \le \frac{1}{{10BW}}$, in which $d_{max}$ denotes the maximum distance  between any two  transmitting antennas, $C$ is the speed of light and $BW$ is the signal bandwidth.
\end{property}

As referred in \cite{Masood2015}, if $\frac{{{d_{\max }}}}{C} \le \frac{1}{{10BW}}$, it is safe to assume that all the transmit-receive antenna pairs have the same nonzero positions of the channel taps in a large-scale MIMO system. We present the parameters of the concerned systems in Table I \cite{Masood2015}. We can get the conclusion that in the LTE system \cite{Rangan2014}, 25$\times$25 array ($24d<d_{max}$, $d$ represents the distance between two adjacent antennas and $\lambda $ is the wavelength) has common channel support. In the mmWave large-scale MIMO system \cite{Rangan2014} proposed for the 5G,  10$\times$10 array ($9d<d_{max}$)  guarantees the common channel sparsity.

\begin{table}
\centering
\caption{Parameters}
\begin{tabular}{|c|c|c|c|c|}
\hline\hline
System&BW&Center frequency&$d_{max}$&$d=\lambda /2$\\
\hline\hline
LTE&20MHz&2.6GHz&1.5m&0.058m\\
\hline
mmWave&1GHz&60GHz&0.03m&0.0025m\\
\hline\hline
\end{tabular}
\end{table}

\newtheorem{theorem}{Theorem}
\begin{theorem}
 The elements of  the BEM coefficients set $\{ {\bm{\tilde {\theta}}} _d^{(n_t)}\}$  $(n_t \in [1, N_t]$, $d \in [0,D-1])$ in a large-scale MIMO system share the common sparsity under the condition of $\frac{{{d_{\max }}}}{C} \le \frac{1}{{10BW}}$, in which $d_{max}$ denotes the maximum distance  between any two  transmitting antennas, $C$ is the speed of light and $BW$ is the signal bandwidth.
\end{theorem}

\begin{proof}
 As ${\mathbf{h}^{(n_t)} _l} ={\mathbf{V}}{{\bm{\theta}}^{(n_t)} _l}$ and ${{\mathbf{h}}^{(n_t)} _l}=\textbf{0}$ for $l \notin \left\{ {l_1}, \cdots ,{l_K} \right\}$, we have ${{\bm{\theta}}^{(n_t)} _l}=\textbf{0}$ for $l \notin \left\{ {l_1}, \cdots ,{l_K} \right\}$. Similar with $\{{\widetilde {\bf{h}}^{(n_t)}_n}\}$,  $\{{\widetilde {\bm{\theta}}^{(n_t)}_d}\}$ are  $K$-sparse vectors and their common  non-zero positions are $\left\{{l_1}, \cdots ,{l_K} \right\}$. As analyzed above, all the transmit-receive antenna pairs have the common sparsity in delay domain in the concerned systems. We have that $\{{\widetilde {\bm{\theta}}^{(n_t)}_d}\}$ share the common sparsity for $d \in [0,D-1]$, $n_t \in [1,N_t]$.

\end{proof}

\subsection{Channel Estimation Scheme}

\begin{figure}[!hbp]
%\vspace{-0.3cm}
%\centering
\includegraphics[scale=0.5]{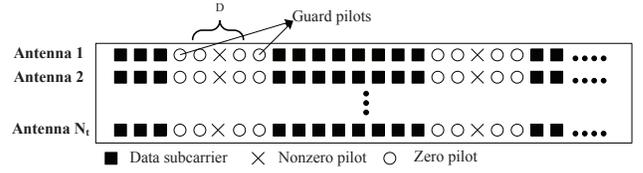}
\caption{Pilot position diagram (for $D$=3)}
\label{1}
\end{figure}

The complex exponential-BEM (CE-BEM) \cite{Tang2011} with order $D$ is exploited in the proposed scheme and the basis function is ${{\bf{v}}_d} = {(
1, \cdots, {{e^{j\frac{{2\pi }}{N}n(d - \frac{{D - 1}}{2})}}}, \cdots, {{e^{j\frac{{2\pi }}{N}(N - 1)(d - \frac{{D - 1}}{2})}}}
 )^T}, d {\in} [0,D{-}1]$.

 We arrange the pilot subcarriers as shown in Fig. 1.  A non-zero pilot  is equipped with $D-1$  zero pilots on each side. The pilot position of all the transmitting antennas is the same.  Assume that the pilots are arranged in $G$ groups. We select the $D$ pilot subcarriers in the middle of each group for channel estimation and the remaining ones are guard pilots. As analyzed in \cite{Tang2011}, it is a ICI-free structure. The index set of the selected subcarriers  $\mathcal{S}_d$  $(d \in [0,D-1])$ includes the $(d+1)$-th selected pilot subcarrier in each group. The index set of the non-zero pilots $\mathcal{S}_{\frac{{D - 1}}{2}}$ is also denoted as $\mathcal{S}_{cen}$. Thus we have
 \begin{equation}
  \begin{array}{*{20}{c}}
{{{\mathcal{S}}_0} = {{\mathcal{S}}_{cen}} - \frac{{D - 1}}{2}}\\
 \vdots \\
{{{\mathcal{S}}_{\frac{{D - 1}}{2}}} = {{\mathcal{S}}_{cen}}}\\
 \vdots \\
{{{\mathcal{S}}_{D - 1}} = {{\mathcal{S}}_{cen}} + \frac{{D + 1}}{2}}
\end{array}.
\end{equation}

\begin{theorem}
 Assume that CE-BEM with order $D$ is utilized to approximate the channel coefficients and ${\mathbf{P}^{{(n_t)}}}$ $(n_t \in [1,N_t])$  denotes the values of the non-zero pilots.  The received selected pilots whose indices correspond to $\mathcal{S}_d$ $(d \in [0,D-1])$ can be expressed as
\begin{equation}
{\left[ {\bf{Y}} \right]_{{{\cal S}_d}}} = \sum\limits_{{n_t} = 1}^{{N_t}} {diag( {{\bf{P}}^{{(n_t)}}}) {{[{{\bf{W}}_L}]}_{{{\cal S}_{cen}}}}} {\bm{\widetilde {\theta}}}_d^{{(n_t)}} + {{{\bm{\eta }}}_d}.
\end{equation}
The recovery of the BEM coefficients $\mathbf{X}$ can be organized as a DCS problem
\begin{equation}\label{DCSproblem}
{\bf{Y}}_R=\bf{\Phi}\bf{X}+{{\bm{\eta }}_R},
\end{equation}
in which
\begin{equation}\label{Y_R}
{{\bf{Y}}_R} = \left( {\begin{array}{*{20}{c}}
{{{\left[ {\bf{Y}} \right]}_{{{\mathcal{S}}_0}}}},& \cdots, &{{{\left[ {\bf{Y}} \right]}_{{{\mathcal{S}}_{D - 1}}}}}
\end{array}} \right),
\end{equation}
\begin{equation}\label{fi}
{\bf{\Phi }} = \left( {diag( {{\bf{P}}^{(1)}}) {{\left[ {{{\bf{W}}_L}} \right]}_{{{\mathcal{S}}_{cen}}}}} \right., \cdots, \left. {diag( {{\bf{P}}^{({N_t})}}) {{\left[ {{{\bf{W}}_L}} \right]}_{{{\mathcal{S}}_{cen}}}}} \right),
\end{equation}
\begin{equation}
{\bf{X}} = \left( {\begin{array}{*{20}{c}}
{{\bm{\widetilde {\theta}}}_0^{(1)}}& \cdots &{{\bm{\widetilde {\theta}}}_{D - 1}^{(1)}}\\
 \vdots & \ddots & \vdots \\
{{\bm{\widetilde {\theta}}}_0^{({N_t})}}& \cdots &{{\bm{\widetilde {\theta}}}_{D - 1}^{({N_t})}}
\end{array}} \right),
\end{equation}
\begin{equation}
{{\bm{\eta}} _R} = \left( {\begin{array}{*{20}{c}}
{{\bm{\eta} _0}},& \cdots, &{{{\bm{\eta}} _{D - 1}}}
\end{array}} \right).
\end{equation}
In (\ref{DCSproblem}), ${\bf{Y}}_R$ is the received selected pilot subcarriers.   In (\ref{Y_R}), $\mathbf{Y}$ denotes the received vector in frequency domain, and ${\left[ \mathbf{Y} \right]_{{\mathcal{S}}_d}}$ consists of the extracted rows of $\mathbf{Y}$ whose indices correspond to the elements in set ${\mathcal{S}_d}$. In (\ref{fi}), $\mathbf{W}_L$  denotes the first $L$ columns of the DFT matrix.
\end{theorem}

\begin{proof}
Please refer to the Appendix.
\end{proof}

  Generate random sequences  consisting of $\pm 1$ with the probability of 1/2  as the values of the non-zero pilots ${\mathbf{P}}^{(n_t)} \in {\mathbb{Z}^{G \times 1}}$, $n_t \in [1,N_t]$. Then generate  ${\mathcal{S}_{cen}}=\{s_1, s_2, \ldots, s_G \}$  randomly, which satisfies $\left| {s_u - s_v} \right| \ge 2D - 1, (s_u, s_v \in [0,N-1], u, v \in [1,G])$.
The measurement matrix $\mathbf{\Phi}$   formulated as (\ref{fi}) has a low MCP to guarantee the recovery performance.

As discussed above, $\{\widetilde{\bm{\theta}}_d^{({n_t})}\}$  $(d \in [0,D-1], n_t \in [1,N_t])$ share the common support.  Apparently, the columns of $\bf{X}$ have common sparsity.  SOMP  is utilized to solve (\ref{DCSproblem}). The CE-BEM coefficients ${\bar{\bm{\theta}}^{({n_t})}}$ $(n_t \in [1,N_t])$ can be obtained by some simple rearrangement  of $\mathbf{X}$. Then substituting  ${\bar{\bm{\theta}}^{({n_t})}}$ $(n_t \in [1,N_t])$ to (\ref{recovery}), we get the channel coefficients ${\bar{\bf h}}^{({n_t})}$ $(n_t \in [1,N_t])$.

\subsection{Linear Smoothing Method for Large-Scale MIMO}

As referred in \cite{Mostofi2005}, if the normalized Doppler shift $\vartheta < 0.2$, the channel coefficients ${\bf{h}}_l^{({n_t})} \in \mathbb{C}^{N \times 1}$ presents linear correlation with instant $n$ $(n \in [0,N-1])$. We utilize this property to process the the channel coefficients  and get a more accurate estimation.

\begin{itemize}
  \item Get the indices of strong taps $\left\{ {l_1}, \cdots ,{l_K} \right\}$ by comparison of the estimated channel coefficients ${\bf{h}}_l^{({n_t})}$, $l \in \left\{ {l_1}, \cdots ,{l_K} \right\}$, $n_t \in [1,N_t]$.
  \item Get two special point of each ${\bf{h}}_{l_k}^{({n_t})}$: $\hat h^{({n_t})}\left[\frac{N}{4}-1,l_k\right] {\approx} \frac{2}{N}(\sum\limits_{n = 0}^{N/2 - 1} {h^{({n_t})}\left[n,{l_k}\right]} )$, $\hat h^{({n_t})}\left[\frac{3N}{4}-1,{l_k}\right] {\approx} \frac{2}{N}(\sum\limits_{n = N/2}^{N-1} {h^{({n_t})}\left[n,{l_k}\right]} )$, $l \in \left\{ {l_1}, \cdots ,{l_K} \right\}$, $n_t \in [1,N_t]$.
  \item Calculate the slope of the approximate line determined by the two points of each ${\bf{h}}_{l_k}^{({n_t})}$: $\beta _{{l_k}}^{({n_t})} = \frac{{\hat h^{({n_t})}\left[\frac{N}{4}-1,{{l_k}}\right] - \hat h^{({n_t})}\left[\frac{3N}{4}-1,{{l_k}}\right]}}{{N/2}}$, $l \in \left\{ {l_1}, \cdots ,{l_K} \right\}$, $n_t \in [1,N_t]$.

  \item The processed channel is expressed as $\hat{h}^{({n_t})}\left[n,{l_k}\right] = (n + 1 - \frac{N}{4})\beta _{{l_k}}^{({n_t})} + \hat h^{({n_t})}\left[\frac{N}{4}-1,{l_k}\right]$, $l \in \left\{ {l_1}, \cdots ,{l_K} \right\}$, $n_t \in [1,N_t]$, $n \in [0,N-1]$ .
\end{itemize}

\section{SIMULATION RESULTS}

In this section, we compare the  normalized mean square error (NMSE) performance of the conventional LS estimator \cite{Tang2011}, the LCC channel estimation scheme \cite{Ren2013}, and our proposed channel estimation scheme under the DS channel in large-scale MIMO systems. Jake's model is employed to generate the DS channel. Quadrature phase-shift keying (QPSK) is adopted as the modulation technique. NMSE(dB)$ = 10lo{g_{10}}(\frac{{\left\| {{{vec({\bf{\hat{h}}})}}} - {vec({\bf{h}})} \right\|_2^2}}{\left\| vec({\bf{h}}) \right\|_2^2})$,  ${{\bf{\hat{h}}}}$ represents the estimated CSI and  $\bf{h}$ stands for the accurate CSI.
\begin{figure}
\centering
\includegraphics[scale=0.4]{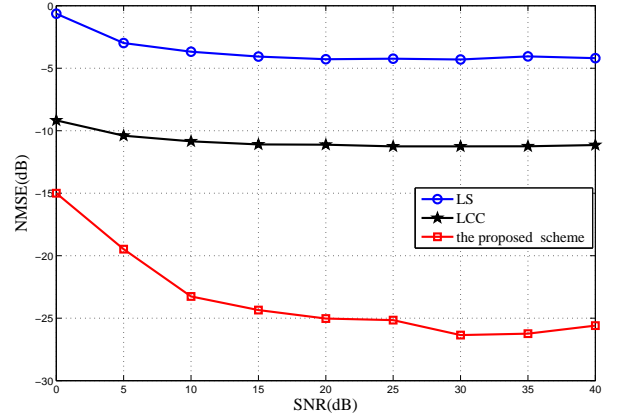}
\label{figure_NMSE}
\caption{The NMSE of the estimated CSI vs. SNR without smoothing.}
\end{figure}
\begin{figure}
\centering
\includegraphics[scale=0.4]{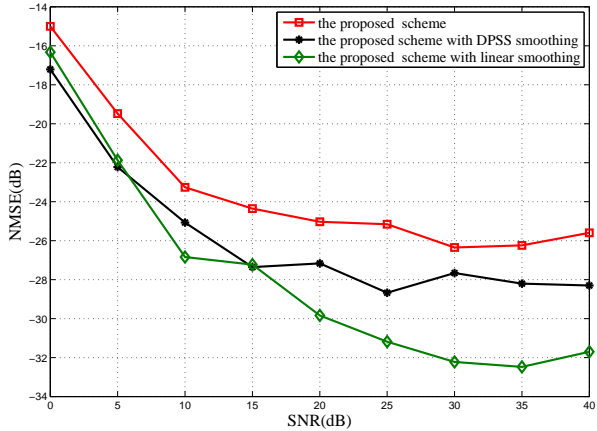}
\label{figure_optimazition}
\caption{The NMSE of the estimated CSI vs. SNR with different smoothing methods.}
\end{figure}
%\begin{figure}
%\centering
%\includegraphics[scale=0.4]{figure_miu.eps}
%\label{figure_miu}
%\caption{The performance comparison between the proposed CDM scheme with and without pilot design.}
%\end{figure}
\begin{figure}
\centering
\includegraphics[scale=0.4]{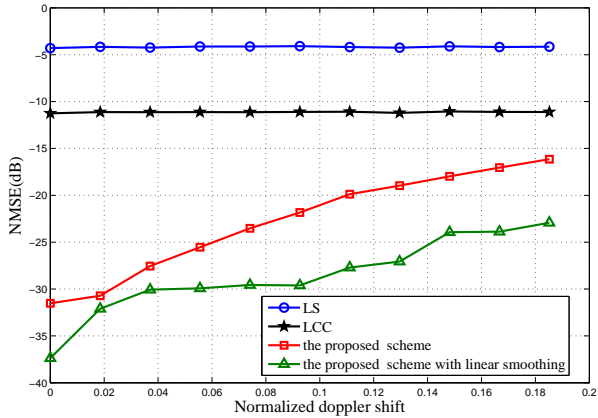}
\label{figure_Doppler}
\caption{The NMSE of the estimated CSI vs. the normalized doppler shift.}
\end{figure}
\begin{figure}
\centering
\includegraphics[scale=0.4]{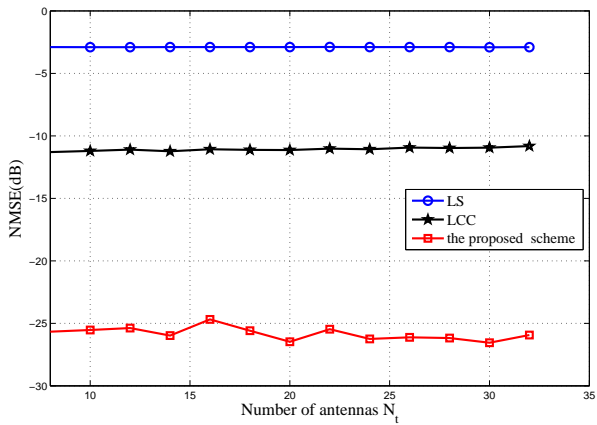}
\label{figure_antenna_cut}
\caption{The NMSE of the estimated CSI vs. the number of antennas $N_t$.}
\end{figure}

In Fig. 2, we compare the NMSE of the estimated CSI versus the  signal-to-noise ratio (SNR), under the number of subcarriers $N=1024$, the number of nonzero pilots $G=96$, the length of the channel $L=16$, the number of the strong paths $K=2$, the CE-BEM order $D=3$, the  normalized  doppler  $\upsilon =0.057 $  and the number of  transmitting antennas $N_t=16$. The total number of pilots is $G(2D-1)=480$. For fair comparison, we assume the same pilot overheads in each estimation scheme. From Fig. 2, we can see that the performance of LCC deteriorates notably  since it  utilizes the sparsity in delay-doppler domain, which is reduced by the large doppler shift and the  large number of antennas. The proposed channel estimation scheme has a substantial performance gain for the reason that  it takes advantage of the common sparsity of all the transmitting antennas in delay domain and also benefits from the  ICI-free structure.

In Fig. 3, we compare the NMSE performance of DPSS smoothing treatment in \cite{Cheng2013} and the proposed linear smoothing method under $N=1024$, $G=96$, $L=16$, $K=2$, $D=3$, $N_t=16$, and $\upsilon =0.057$. From this figure, we can see the performance gain of the linear smoothing method scheme. This is because it takes full advantage of the linear characteristics presented by the DS channels with the normalized doppler shift less than 0.2.

%In Fig. 4, we verify the role of the pilot design algorithm under $N=1024$, $G=96$, $L=16$, $K=2$, $Q=3$, $N_t=16$, and $\upsilon =0.057$. We can see that our proposed  scheme is ineffective without pilot design algorithm for the reason that the same columns appear in the measurement matrix and results in $\mu=1$. Thus the recovery performance can not be guaranteed  according to the basic theory of CS.

In Fig. 4, we compare the NMSE ¡¡performance versus the normalized doppler shift $\upsilon$ under $N=1024$, $G=96$, $L=16$, $K=2$, $D=3$, $N_t=16$, and SNR=20dB.  From this figure, we observe that the quality of the estimated CSI gets worse as the normalized doppler shift
$\upsilon$ increases. This is because a large $\upsilon$ leads to a large channel modeling error.

In Fig. 5, we compare the NMSE ¡¡performance versus the number of antennas $N_t$, under $N=1024$, $L=16$, $K=2$, $D=3$, and SNR=20dB. The number of pilots is $3KN_t(2D-1)$, which increases in proportion to $N_t$. From this figure, we observe that the variation of $N_t$ nearly has no effect on the performance of the estimated CSI. This is because the more pilot overheads make up for the performance deterioration brought by the increase of $N_t$.  Our proposed scheme can be applied to the situations with different number of antennas.

\section{Conclusion}
In this paper,  a compressive channel estimation scheme is proposed  for DS channel in large scale MIMO systems. In this scheme, we introduce the BEM to reduce the number of  coefficients to be estimated, which decreases from $N_tNL$ to $N_tDL$, $D \ll  N$. At the same time, the requirement of the pilot overheads  is also decreased. Then we analyze the  sparsity of the BEM coefficients  of all the transmit-receive antenna pairs in delay domain. The BEM coefficients estimation is formulated  as a DCS problem, which has a linear structure with low complexity. Moreover, the proposed linear smoothing method improves the accuracy of estimation.  So we solve the problem of unaffordable pilot overheads and prohibitive complexity for DS channel estimation in large scale MIMO systems. Simulation results verify the accuracy of the proposed scheme.

\appendix[Proof of Theorem 2]

As derived in (\ref{recfre}), the received vector in frequency domain is
\begin{equation}
{\bf{Y}} = \sum\limits_{{n_t} = 1}^{{N_t}} {{\bf{H}}_f^{({n_t})}{\bf{S}}^{({n_t})} + {\bf{E}}}.
\end{equation}
Substitute (\ref{matfre}) into (\ref{recfre}), we have
\begin{equation}
{\bf{Y}} = \sum\limits_{{n_t} = 1}^{{N_t}} {\sum\limits_{d = 0}^{D - 1} {{{\bf{V}}_d}{{\bf{\Theta}}^{(n_t)}_d}{\bf{S}}^{({n_t})} + {\bf{E}}} }.
\end{equation}
${{\bf{V}}_d}$ in CE-BEM can be simplified as
\begin{equation}\label{Bq}
\begin{split}
{{\bf{V}}_d} &= {{\bf{W}}}diag( {{\bf{v}}_d}) {{\bf{W}}^H}\\
& = {\bf{E}}_N^{{ \downarrow _\alpha }}{{\bf{W}}}{{\bf{W}}^H}\\
& = {\bf{E}}_N^{{ \downarrow _\alpha }},
\end{split}
\end{equation}
in which, the CE-BEM function ${{\bf{v}}_d} = {\left( {\begin{array}{*{20}{c}}
1,& \cdots, &{{e^{j\frac{{2\pi }}{N}n(d - \frac{{D - 1}}{2})}}},& \cdots, &{{e^{j\frac{{2\pi }}{N}(N - 1)(d - \frac{{D - 1}}{2})}}}
\end{array}} \right)^T}$, ${{\bf{W}}}$ is the DFT matrix of order $N$, and ${\bf{E}}_N^{{ \downarrow _\alpha }}$ denotes that the $N$-order identity matrix shifts down circularly by $\alpha$ rows, $\alpha=d - \frac{{D - 1}}{2}$. ${{\bm{\Theta}}^{(n_t)}_d}{\bf{S}}^{({n_t})}$ can be simplified as
\begin{equation}\label{CqS}
\begin{split}
{{\bm{\Theta}}^{(n_t)}_d}{\bf{S}}^{({n_t})} & = diag( \sqrt N {{\bf{W}}}{\left( {{{\left( {{\widetilde{{\bm{\theta}}}^{(n_t)}_d}} \right)}^T},{0_{1 \times (N - L)}}} \right)^T}) {\bf{S}}^{({n_t})}\\
 &= diag( {\bf{S}}^{({n_t})}) {{\bf{W}}_{L}}{\widetilde{{\bm{\theta}}}^{(n_t)}_d},
\end{split}
\end{equation}
in which ${{\bf{W}}_{L}}$ represents the first $L$ columns of ${\bf{W}}$.
Substituting (\ref{Bq}) and (\ref{CqS}) into (\ref{recfre}), we can get that
\begin{equation}
{\bf{Y}} = \sum\limits_{{n_t} = 1}^{{N_t}} {\sum\limits_{d = 0}^{D - 1} {{\bf{E}}_N^{{ \downarrow _\alpha }} diag{({\bf{S}}^{({n_t})})} {{\bf{W}}_L}{\widetilde{{\bm{\theta}}}}_d^{({n_t})}} }  + {\bm{\eta }}.
\end{equation}
Selecting the pilot subcarriers, we have
\begin{equation}\label{selectrec}
\begin{split}
{{\bf{U}}_{\bar d}}{\bf{Y}}& = {{\bf{U}}_{\bar d}}(\sum\limits_{{n_t} = 1}^{{N_t}} {\sum\limits_{d = 0}^{D - 1} {{\bf{E}}_N^{{ \downarrow _\alpha }}diag{({\bf{S}}^{({n_t})})}{{\bf{W}}_L}\widetilde{{\bm{\theta}}}_d^{({n_t})}} } ) + {{\bf{U}}_{\bar d}}{\bm{\eta }}\\
 &= \sum\limits_{{n_t} = 1}^{{N_t}} {\sum\limits_{d = 0}^{D - 1} {{{\bf{U}}_{\bar d}}{\bf{E}}_N^{{ \downarrow _\alpha }}diag{({\bf{S}}^{({n_t})})}{{\bf{W}}_L}\widetilde{{\bm{\theta}}}_d^{({n_t})} + {{\bf{U}}_{\bar d}}{\bm{\eta }}} },
 \end{split}
\end{equation}
in which, ${{\bf{U}}_{\bar d}}$ consists of the selected rows of the $N$-order identity matrix whose indices correspond to the elements of $\mathcal{S}_{\bar d}$, ${\bar d} \in [0,D-1]$. It is not hard to discover that
\begin{equation}\label{select}
{{\bf{U}}_{\bar d}}{\bf{E}}_N^{{ \downarrow _\alpha }}diag{({\bf{S}}^{({n_t})})} = \left\{ {\begin{array}{*{20}{c}}
{diag( {{\bf{P}}^{({n_t})}})[\mathbf{I}_N]_{{\mathcal{S}_{cen}}}}\\
{\bf{0}}
\end{array}} \right.{\rm{ }}\begin{array}{*{20}{c}}
{\bar d = d = \frac{D-1}{2}}\\
{else}
\end{array}.
\end{equation}
Substitute (\ref{select}) into (\ref{selectrec}), we can get
\begin{equation}
{[{\bf{Y}}]_{{\mathcal{S}_d}}} = \sum\limits_{{n_t} = 1}^{{N_t}} {diag( {{\bf{P}}^{({n_t})}}) {{[{{\bf{W }}_L}]}_{{\mathcal{S}_{cen}}}}\widetilde{{\bm{\theta}}}_d^{({n_t})} + {{\bm{\eta }}_d}},
\end{equation}
in which, ${\left[ \mathbf{Y} \right]_{{\mathcal{S}}_d}}$ consists of the extracted rows of $\mathbf{Y}$ whose indices correspond to the elements in set ${\mathcal{S}_d}$.

We express ${\left[ \mathbf{Y} \right]_{{\mathcal{S}}_d}}$ in another way
\begin{equation}
\begin{split}
\begin{array}{l}
\begin{array}{*{20}{l}}
{{{[{\bf{Y}}]}_{{{\mathcal{S}}_d}}} = \left( {diag({{\bf{P}}^{(1)}}){{\left[ {{{\bf{W}}_L}} \right]}_{{{\mathcal{S}}_{cen}}}}} \right., \cdots ,\left. {diag({{\bf{P}}^{({N_t})}}){{\left[ {{{\bf{W}}_L}} \right]}_{{{\mathcal{S}}_{cen}}}}} \right)}\\
{\quad \quad \quad \quad  \cdot \left( {\begin{array}{*{20}{c}}
{{\bm{\widetilde \theta }}_d^{(1)}}\\
 \vdots \\
{{\bm{\widetilde \theta }}_d^{({N_t})}}
\end{array}} \right) + {{\bm{\eta }}_d}}
\end{array}\\
\quad \quad \quad = {\bm{\Phi }}\left( {\begin{array}{*{20}{c}}
{{\bm{\widetilde \theta }}_d^{(1)}}\\
 \vdots \\
{{\bm{\widetilde \theta }}_d^{({N_t})}}
\end{array}} \right) + {{\bm{\eta }}_d},
\end{array}
\end{split}
\end{equation}
in which, ${\bf{\Phi }} {=} \left( {diag({{\bf{P}}^{(1)}}){{\left[ {{{\bf{W}}_L}} \right]}_{{{\mathcal{S}}_{cen}}}}},\cdots, {diag({{\bf{P}}^{({N_t})}}){{\left[ {{{\bf{W}}_L}} \right]}_{{{\mathcal{S}}_{cen}}}}} \right)$, $d \in [0,D-1]$. All the received selected pilot subcarriers
\begin{equation}
\begin{split}
{{\bf{Y}}_R} &= \left( {\begin{array}{*{20}{c}}
{{{\left[ {\bf{Y}} \right]}_{{{\mathcal{S}}_0}}}},& \cdots, &{{{\left[ {\bf{Y}} \right]}_{{{\mathcal{S}}_{D - 1}}}}}
\end{array}} \right)  \\
&= \mathbf{\Phi} \left( {\begin{array}{*{20}{c}}
{{\bm{\widetilde {\theta}}}_0^{(1)}}& \cdots &{{\bm{\widetilde {\theta}}}_{D - 1}^{(1)}}\\
 \vdots & \ddots & \vdots \\
{{\bm{\widetilde {\theta}}}_0^{({N_t})}}& \cdots &{{\bm{\widetilde {\theta}}}_{D - 1}^{({N_t})}}
\end{array}} \right)+ {{\bm{\eta }}_R} \\
&= \mathbf{\Phi}\mathbf{X}+ {{\bm{\eta }}_R},
\end{split}
\end{equation}
in which, ${{\bm{\eta}} _R} = \left( {\begin{array}{*{20}{c}}
{{\bm{\eta} _0}},& \cdots, &{{{\bm{\eta}} _{D - 1}}}
\end{array}} \right)$ and ${\bf{X}} = \left( {\begin{array}{*{20}{c}}
{{\bm{\widetilde {\theta}}}_0^{(1)}}& \cdots &{{\bm{\widetilde {\theta}}}_{D - 1}^{(1)}}\\
 \vdots & \ddots & \vdots \\
{{\bm{\widetilde {\theta}}}_0^{({N_t})}}& \cdots &{{\bm{\widetilde {\theta}}}_{D - 1}^{({N_t})}}
\end{array}} \right)$. As the elements of the set $\{\widetilde{\bm{\theta}}_d^{({n_t})}\}$  $(d \in [0,D-1], n_t \in [1,N_t])$ have common sparsity, the columns of $\mathbf{X}$ share the same support apparently. We have that ${\bf{Y}}_R=\bf{\Phi}\bf{X}+{{\bm{\eta}} _R}$  is a DCS problem.
\ifCLASSOPTIONcaptionsoff
  \newpage
\fi

\bibliographystyle{IEEEtran}
\bibliography{IEEEabrv,IEEEfull}

\end{document}